\documentclass[runningheads]{llncs}

\usepackage{makeidx}

\usepackage{moreverb} 
\usepackage{pdfpages}
\usepackage{amsmath}

\usepackage{graphicx}
\usepackage{amsmath}
\usepackage{fancybox}
\usepackage{tikz}
\usepackage{fancyhdr}
\usepackage{amsfonts}
\usepackage{amssymb}

\usepackage{lscape}
\usepackage[utf8]{inputenc}

\usepackage{listings}

\usepackage{multirow}
\usepackage{multicol}

\usepackage{graphicx} 

\usepackage{url}
\definecolor{Blue}{rgb}{0.1,0.1,0.9}
\definecolor{Red}{rgb}{0.9,0.1,0.1}

\newcommand{\overarrow}[1]{\stackrel{#1}{\rightarrow}}

\makeatletter
\DeclareRobustCommand\varname[1]{%
\ifmmode
  \begingroup %
  \let\math@bgroup\@empty \let\math@egroup\@empty \mathit\@empty #1\endgroup
\else
  \textit{#1}%
\fi}
\makeatother

\newcommand{\raw}{\rightarrow}

\newcommand{\Tc}{{T}}

\newcommand{\pos}{{\cP}os}

\newcommand*{\Symbols}{{\Sigma}}

\newcommand*{\Variables}{\cV}
\newcommand*{\Var}[1]{\mathit{Var}{#1}}

\newcommand*{\cO}{{\cal O}}

\def\defemb#1#2{\expandafter\def\csname #1\endcsname
                              {\relax\ifmmode #2\else\hbox{$#2$}\fi}}

\def\ll{[\![}
\def\rr{]\!]}
\def\Den#1{\relax\ifmmode \ll #1\rr \else\hbox{$\ll #1\rr$}\fi}

\let\|=\mid

\def\#{\hat{~}}

\defemb{implies}{\quad\rightarrow\quad}
\defemb{Implies}{\quad\Rightarrow\quad}

\defemb{cA}{{\cal A}}
\defemb{cB}{{\cal B}}
\defemb{cC}{{\cal C}}
\defemb{cD}{{\cal D}}
\defemb{cE}{{\cal E}}
\defemb{cF}{{\cal F}}
\defemb{cG}{{\cal G}}
\defemb{cH}{{\cal H}}
\defemb{cI}{{\cal J}}
\defemb{cJ}{{\cal J}}
\defemb{cL}{{\cal L}}
\defemb{cP}{{\cal P}}
\defemb{cR}{{\cal R}}
\defemb{cS}{{\cal S}}
\defemb{cT}{{\cal T}}
\defemb{cU}{{\cal U}}
\defemb{cV}{{\cal V}}
\defemb{cZ}{{\cal Z}}

\DeclareSymbolFont{boldsymbols}{OMS}{cmsy}{b}{n}
\DeclareSymbolFontAlphabet{\mathbfcal}{boldsymbols}

\long\def\comment#1{}

\newcommand{\startprog}{\begin{prog}}
\newcommand{\stopprog}{\end{prog}\noindent}

\begin{document}

\pagestyle{headings}

\title{Dynamic Backward Slicing of \\ Rewriting Logic Computations\thanks{This work has been partially supported by the EU (FEDER) and the Spanish MEC  TIN2010-21062-C02-02 project, 
by Generalitat Valenciana PROMETEO2011/052,
and by the Italian MUR under grant RBIN04M8S8, FIRB project, Internationalization 2004.
Daniel Romero is also supported by FPI--MEC grant BES--2008--004860.
}}\titlerunning{Dynamic Backward Slicing of Rewriting Logic Computations}  %

\author{M. Alpuente\inst{1} \and D. Ballis\inst{2} \and J. Espert\inst{1} \and D. Romero\inst{1}}
\authorrunning{M. Alpuente, D. Ballis, J. Espert, and D. Romero} %

\institute{
DSIC-ELP, Universidad Polit\'ecnica de Valencia\\
Camino de Vera s/n, Apdo 22012, 46071 Valencia, Spain\\
\email{\{alpuente,jespert,dromero\}@dsic.upv.es}
\and
Dipartimento di Matematica e Informatica\\
Via delle Scienze 206, 33100 Udine, Italy\\
\email{demis.ballis@uniud.it}
}

\maketitle

\begin{abstract}

Trace slicing is a widely used technique for execution trace 
analysis that is effectively used in program debugging, analysis and comprehension. 
In this paper, we present a 
backward trace slicing technique that can be used  for the analysis of Rewriting Logic theories.

Our trace slicing technique allows us to systematically 
trace back rewrite sequences modulo equational axioms (such as associativity and commutativity)
by means of an  algorithm that dynamically simplifies the traces by detecting control and data dependencies, and dropping useless data  that do not influence the final result.  
Our methodology is particularly  suitable for analyzing complex, textually-large system computations 
such as those delivered  as counter-example traces by 
Maude model-checkers.
\end{abstract}

\section{Introduction}\label{sec:intro}

The analysis of execution traces plays a fundamental role in many program manipulation techniques.
Trace slicing is a technique for 
reducing the size of traces by focusing on selected aspects of program execution, which makes it suitable for trace analysis
and monitoring~\cite{ChenR09}.

Rewriting Logic (RWL) is a very general \emph{logical} and \emph{semantic framework},
which is particularly suitable for formalizing  highly concurrent, complex systems 
(e.g., biological systems \cite{BBF09,Tal08} and Web systems \cite{ABER10,ABR09}). RWL is efficiently implemented in the high-performance system 
Maude \cite{maude-book}. 
Roughly speaking, a {\em rewriting logic  theory} seamlessly combines  
a {\em term rewriting system} (TRS) together with an {\em equational theory} that may include sorts, functions, and algebraic laws (such as commutativity and associativity)  so that rewrite steps are applied {\em modulo} the equations.
Within this framework, the system states are typically represented as elements of an algebraic data type that is specified by the equational theory, while the system computations  are  modeled via the rewrite rules, which describe  transitions between states. 

Due to the many important applications of RWL,  in recent years, the debugging and optimization of RWL theories 
have received growing attention 
\cite{ABBF10,MM02,RVM2010}. 
However, 
the existing tools provide hardly support for execution
trace analysis.
The original motivation for our work was to reduce the size of the counterexample traces  delivered by Web-TLR, which is a RWL-based model-checking tool for Web applications 
proposed in  \cite{ABER10,ABR09}.
As a matter of fact, the analysis (or even the simple inspection) of such traces 
may be unfeasible because of the size and complexity of the traces under examination.
Typical counterexample traces in  Web-TLR are {75~Kb} long
for a model size of {1.5~Kb}, that is, the trace is in a ratio of 5.000\% w.r.t.\ the model.

To the best of our knowledge, this paper presents the first 
trace slicing technique for RWL theories. The basic idea is to take a trace produced by the RWL engine and traverse and analyze it backwards to filter out events that are irrelevant for the rewritten task. 
The   trace slicing technique that we propose is fully general and can be applied to 
optimizing any RWL-based tool that manipulates  rewrite logic traces. 
Our  technique relies on a suitable mechanism of backward tracing 
that is formalized by means of a procedure that labels the calls (terms) involved in the rewrite steps. 
This allows us to infer, from a term $t$ and positions of interest on it,  positions of interest of the term that was rewritten to $t$.
Our labeling procedure extends the  technique in~\cite{BKV00}, which allows descendants and origins to be traced in orthogonal (i.e.,  left-linear and overlap-free) term rewriting systems in order  to deal with rewrite theories that may contain commutativity/associativity axioms, 
as well as  nonleft-linear, collapsing equations and rules.
As in dynamic tracing  \cite{FT94,Terese03book}, our definition of labeling uses a relation on contexts derived from the reduction relation on terms, where the symbols in the left-hand side 
of a rule propagate to all symbols of its right-hand side.
This   labeling relation allows us to make precise the {\em dynamic dependence} of function symbols  occurring in the terms of a reduction sequence on symbols in previous terms in that sequence~\cite{FT94}.

\medskip

\noindent{\em Plan of the paper.}
Section~\ref{sec:prelim} summarizes some preliminary definitions and notations about term rewriting systems.
In Section~\ref{sec:rewrite-m-e}, we recall the essential notions concerning  
rewriting modulo  equational theories.
Section~\ref{sec:lab-trac}  describes the main kinds of labeling and tracing in term rewrite systems. 
In Section~\ref{sec:labeling_and_slicing}, we formalize our backward trace slicing technique for elementary rewriting logic theories.
Section~\ref{sec:slicing-m-e} extends the trace slicing technique of Section~\ref{sec:labeling_and_slicing} by considering extended rewrite theories, i.e., rewrite theories that may include collapsing, nonleft-linear rules, associative/commutative equational axioms, and built-in operators.
Section \ref{sec:exp} describes a software tool that implements the proposed backward slicing technique and reports on an experimental evaluation of the tool that allows us to assess the practical advantages of the trace slicing technique.
In Section~\ref{sec:related}, we discuss some related work and 
then we conclude.
Proofs of the main technical results can be found in Appendix~\ref{app:proof}.

\section{Preliminaries}\label{sec:prelim}

A many-sorted signature $(\Symbols,S)$ consists of a set of sorts $S$ and a $S^* \times S$-indexed family of sets $\Symbols=\{\Sigma_{\bar{s} \times s}\}_{(\bar{s},s)\in S^* \times S}$, which are sets of {\em function symbols} (or operators)  with a
given string of argument sorts and result sort. 
Given an $S$-sorted set  $\Variables = \{\Variables_s\mid s\in S\}$  of disjoint sets of variables, 
$\Tc_{\Sigma}(\Variables)_s$ and ${\Tc_{\Sigma}}_s$ are the sets of terms and ground terms of sorts $s$, respectively. We write $\Tc_{\Sigma}(\Variables)$ and $\Tc_{\Sigma}$ for the corresponding term algebras. 
An \emph{equation} is a pair of terms of the form $s=t$, with $s,t\in \Tc_{\Sigma}(\Variables)_s$.  
In order to simplify the presentation, we often disregard %
sorts when no confusion can arise.

Terms are viewed as labelled trees in the usual way.  
Positions are
represented by sequences of natural numbers denoting an access path in a
term. The empty sequence $\Lambda$ denotes the root position.
By $root(t)$, we denote the symbol that occurs at the root position of $t$. 
We let $\pos(t)$  denote the set of positions of $t$.
By notation $w_1.w_2$, we denote the concatenation of positions (sequences) $w_1$ and  $w_2$.
Positions are ordered by the prefix ordering, that is, given the positions $w_1,w_2$, 
$w_1\leq w_2$ if there exists a position $x$ such that  $w_1.x=w_2$.
$t_{|u}$ 
is the subterm at the position $u$ of $t$.  $t[r]_u$ is
the term $t$ with the subterm rooted at the position $u$ replaced by $r$.  
A substitution $\sigma$ is a mapping from 
variables to terms  \linebreak $\{x_1 / t_1,\ldots,x_n / t_n\}$  such that $x_i\sigma = t_i$ for $i=1,\ldots,n$
(with $x_{i}\neq x_{j}$ if  $i \neq j$), and $x\sigma = x$ for any
other variable $x$. By $\varepsilon$, we denote the {\em empty} substitution.
Given a substitution $\sigma$, the {\em domain} of $\sigma$ is
the set $\varname{Dom}(\sigma)=\{x|x\sigma\neq x\}$.
By $\Var(t)$ (resp.\ $FSymbols(t)$), we denote the set of variables (resp.\ function symbols) occurring in the term $t$.

A \emph{context} is a term $\gamma\in\Tc_{\Sigma\cup\{\Box\}}(\Variables)$  with zero or
more holes $\Box$\footnote{Actually, when considering types, we assume to have a distinct $\Box_s$ symbol for each sort $s\in S$, and by abuse we  simply denote $\Box_s$ by $\Box$.}, and $\Box\not\in\Sigma$. We write $\gamma[\ ]_u$ to denote that there is a hole at position $u$ of $\gamma$.
By notation $\gamma[\ ]$, we define an arbitrary context (where the number and the positions of the holes
are clarified \emph{in situ}), while we write $\gamma[t_1,\ldots t_n]$ to denote the term obtained by filling the holes
appearing in $\gamma[\ ]$ with terms $t_1,\ldots,t_n$.
By notation $t^\Box$, we denote the context  obtained  by applying the substitution $\sigma = \{ x_1/\Box, \ldots, x_n/\Box\}$ to $t$, where $\Var(t) = \{x_1\ldots,x_n\}$ (i.e., 
$ t^\Box=t\sigma$).
A \emph{term
rewriting system} (TRS for short) is a pair $(\Symbols,{R})$, where
$\Symbols$ is a signature and $R$ is a finite set of reduction (or
rewrite) rules of the form $\lambda \rightarrow \rho$,
$\lambda, \rho \in \Tc_{\Sigma}(\Variables)$, $\lambda \not \in \Variables$ and $\Var(\rho)
\subseteq \Var(\lambda)$.
We often write just $R$ instead of $(\Symbols,{R})$. 
A rewrite step is the application of a rewrite
rule to an expression.  A term $s$
\emph{rewrites} to a term $t$ via $r \in{R} $, $s \overarrow{r }_R t$ (or  $s\stackrel{r ,\sigma}{\rightarrow_R}t$),
if there exists a position $q$ in $s$ such that $\lambda$ {\em matches\/} $s_{|q}$ via 
a substitution $\sigma$ (in symbols,
$s_{|q} = \lambda\sigma$), and $t$ is obtained from $s$ by replacing the subterm $s_{|q} = \lambda\sigma$ with the term $\rho\sigma$, in symbols $t = s[\rho\sigma]_q$. 
The rule $\lambda\raw \rho$ (or equation $\lambda= \rho$) is \emph{collapsing} if $\rho \in \Variables$;
it is \emph{left-linear} if no variable occurs in $\lambda$ more than once.
We denote the transitive and reflexive closure of $\raw$ by $\raw^*$.

Let $r : \lambda \rightarrow \rho$  be a rule.
We call the context $\lambda^\Box$  (resp.~$\rho^\Box$)  {\em redex pattern} (resp.\ {\em contractum pattern}) of $r$. For example, 
the context $f(g(\Box,\Box),a)$ (resp.\ $d(s(\Box),\Box)$) is the redex pattern (resp.\ contractum pattern) of the rule $r : f(g(x,y),a) ) \rightarrow d(s(y),y)$, where $a$ is a constant symbol.

\section{Rewriting Modulo Equational Theories}\label{sec:rewrite-m-e}

An \emph{equational theory}
is a pair $(\Sigma, E)$, where
$\Sigma$ is a signature %
 and $E=\Delta \cup B$ consists of a set of (oriented) equations $\Delta$ together with a collection $B$ of equational axioms (e.g., associativity and commutativity axioms) 
 that are associated with some operator of $\Sigma$. 
The equational theory $E$ induces a least congruence relation on the term algebra $\Tc_{\Sigma}(\Variables)$,
which is usually denoted by $=_E$. 
A {\em rewrite theory} is a triple $\cR = (\Sigma, E, R)$, where $(\Sigma, E)$ is an equational theory, and 
$R$ is a TRS.  Examples of rewrite theories   can be found in \cite{maude-book}.
Rewriting modulo equational theories \cite{MM02} can be defined by lifting the standard  rewrite relation $\raw_{R}$ on terms to the $E$-congruence classes induced by $=_E$. 
More precisely, the rewrite relation $\to_{R/E}$ for rewriting modulo $E$ is defined as $=_E \circ \to_R \circ =_E$. 
A computation in $\cR$ using $\raw_{R\cup\Delta,B}$ is a {\em rewriting logic deduction}, in which the {\em equational simplification} %
with $\Delta$ (i.e., applying the oriented equations in $\Delta$ to a term $t$ until    a canonical form $t\!\downarrow_{E}$ is reached where no further equations can be applied) is intermixed with the rewriting computation with the rules of~$R$, using an {\em algorithm of matching modulo\footnote{A subterm of $t$ matches $l$ ({\em modulo $B$}) via the substitution  $\sigma$ if $t=_B u$ and $u_{|q}=l\sigma$ for a position $q$ of $u$.} $B$} in both cases. 
Formally, given a rewrite theory $\cR = (\Sigma, E, R)$, where $E=\Delta\cup B$, 
a {\em rewrite step modulo $E$}  on a term $s_0$ by means of the rule $r : \lambda\rightarrow\rho \in R$  (in symbols, $s_0 \overarrow{r }_{R\cup\Delta,B}  s_1$) 
can be implemented as follows: 
$(i)$ apply (modulo $B$) the equations of  $\Delta$ %
on $s_0$ to reach a canonical form $(s_0\downarrow_{E})$;
$(ii)$ rewrite (modulo $B$) $(s_0\downarrow_{E})$ to term $v$ by using $r\in R$;
and $(iii)$, apply (modulo $B$) the equations of $\Delta$ %
 on $v$ again to reach a canonical form for~$v$, ${s_1=v\downarrow_{E}}$.

Since the equations of $\Delta$  are implicitly oriented (from left to right), the equational simplification can be seen as a sequence of (equational) rewrite steps ($\rightarrow_{\Delta/B}$).
Therefore, a {\em rewrite step modulo $E$} %
$s_0 \overarrow{r }_{R\cup\Delta,B} s_1$ can be expanded into a sequence of rewrite steps as follows:
{\scriptsize
$$
\begin{array}{c}
  \hspace*{.3cm} \mbox{\scriptsize equational} \hspace*{1.2cm} \mbox{\scriptsize rewrite} \hspace*{1.2cm} \mbox{\scriptsize equational} \hspace*{.5cm} \\ 
 \hspace*{.4cm}   \mbox{\scriptsize simplification} \hspace*{1cm} \mbox{\scriptsize  step$/_B$} \hspace*{1.1cm} \mbox{\scriptsize  simplification} \hspace*{.6cm} \\
 s_0 \overbrace{\rightarrow_{\Delta/B}..\rightarrow_{\Delta/B}   {s_0\!\!\downarrow}_{E}}  \overbrace{=_B u \overarrow{r }_R v} \overbrace{\rightarrow_{\Delta/B} ..\rightarrow_{\Delta/B}  {v\!\!\downarrow}_{E}} = s_1  
\end{array}
$$
}

Given a finite rewrite sequence $\cS=s_0 \raw_{R\cup\Delta,B} s_1 \raw_{R\cup\Delta,B}\ldots  \raw s_n$ in the rewrite theory $\cR$, the \emph{execution trace} of $\cS$ is the rewrite sequence $\cT$ obtained by expanding all the rewrite steps $s_i \raw_{R\cup\Delta,B} s_{i+1}$ of $\cS$ as is described above. %

In this work, a rewrite theory $\cR = (\Sigma, B \cup \Delta, R) $ is called {\em elementary} 
if $\cR$ does not contain equational axioms ($B=\emptyset$) and both rules and equations are left-linear and not collapsing.

\section{Labeling and Tracing in Term Rewrite Systems}\label{sec:lab-trac}

Labeling an object allows us to distinguish it within a collection of identical objects.
This is a useful means to keep track of a given object in a dynamic system.
In the following, we introduce a rather intuitive example that allows us to illustrate how the labeling and tracing process work.

\begin{example}\label{ex:ex-lab}
Let 
${r_1 :  f(x) \raw b}$, and ${r_2 :  g(b) \raw m(a)}$ be two rewrite rules.
Let $g(f(a))$ be an initial term.
Then, by applying $r_1$ and $r_2$ we get the execution trace 
$\cT = g(f(a)) \overarrow{r_1} g(b) \overarrow{r_2} m(a)$.
\end{example}

In term rewriting,  we distinguish three kinds of labeling according to the information recorded by them in an execution trace.

\begin{itemize}

\item
[$(i)$] 
The Hyland--Wadsworth labeling~\cite{Hyland76,Wadsworth76} records the creation level of each symbol.
Roughly speaking, from an initial (default) creation level, the accomplishment of a rewrite step increases by one the creation level of the affected symbols.
For example, consider the execution trace $\cT$ of Example~\ref{ex:ex-lab} together with an initial level $0$ for all symbols. 
Then,
$$g^0(f^0(a^0)) \overarrow{r_1} g^0(b^1) \overarrow{r_2} m^2(a^2) $$

\item
[$(ii)$] The Boudol--Khasidashvili labeling~\cite{Boudol85,Khasidashvili88,Khasidashvili93} records the history of the term in execution traces.
The general idea is to record in the history the applied rule and the symbols of the redex pattern.
This information is taken as the label for the head symbol of the contractum pattern.
Consider again Example~\ref{ex:ex-lab}. 
First, the set of rules is labeled as follows:
$$
r_{1_{f(x)}} :  f(x) \raw r_{1_{f(x)}} \hspace*{1cm}
r_{2_{g(b)}} :  g(b) \raw r_{2_{g(b)}}(a)
$$%
Then, the labeling of the execution trace $\cT$ is:
$$ g(f(a)) \raw g( r_{1_{f(x)}} ) \raw  r_{2_{g(b)}}(a)  $$%
Note that the initial term of this sequence is not labeled, i.e., the initial label is the identity.

\item
[$(iii)$] The Lévy labeling~\cite{Levy76} records the history of each symbol in the term.
Basically, this labeling combines the previous two labelings and attaches the history on every symbol of the contractum pattern. Let us show an example. As before, consider Example~\ref{ex:ex-lab}.
The labeled rules are as follows:
$$
r_{1_{f(x)^\lambda}} : f(x)^\lambda \raw r_{1_{f(x)^\lambda}}^\lambda   \hspace{1cm}
r_{2_{g(b)^\lambda}} : g(b)^\lambda \raw r_{2_{g(b)^\lambda}}^\lambda ( r_{2_{g(b)^\lambda}}^1  )
$$
and the labeled trace of $\cT$ is:
$$
g(f(a))^\lambda ( g(f(a))^1 (g(f(a))^{1.1})  ) \raw
g(f(a))^\lambda ( r_{1_{f(x)^\lambda}}^1 ) \raw 
r_{2_{g(b)^\lambda}}^\lambda ( r_{2_{g(b)^\lambda}}^1 )
$$
Note that due to the accumulation of labels, Lévy labels soon become neither readable nor legible.
Note also that this labeling keeps the maximal information in a rewrite step.

\end{itemize}

In this work, we rely on Klop labeling~\cite{BKV00}, which is inspired by Lévy labeling.
Roughly speaking, Klop labeling employs Greek letters and concatenation of Greek letters as labels.
That is, 
given a rewrite step $t \raw s$, the symbols of $t$ are decorated by using Greek letters as labels.
Then, a new label $l$ is formed by concatenating the labels of the redex pattern. 
Finally, $l$ is attached to every symbol of the contractum pattern of $s$.
A formal definition of this labeling adapted to deal with rewriting logic theories is given in Section~\ref{sec:labeling}.

\medskip

Given a rewrite step $t \raw s$,  tracing allows one to establish a mapping among symbols of $t$ and symbols of $s$. 
Each symbol is mapped according to its location.
For example, occurrences of  symbols in the context of $t$, or in the computed substitution, are traced to the  same occurrences in $s$.
On the contrary, the mapping for the symbols in the redex and contractum patterns depend on the kind of tracing we adopt. 
Namely, in  {\em static} tracing the symbols do not persist through the execution trace.
On other hand, in {\em dynamic} tracing the symbols of the redex pattern are mapped to all symbols of the contractum pattern.
Let us illustrate this by means of an example.

\begin{example}
Consider the  rewrite step $g(f(a)) \overarrow{r_1} g(b)$ into the trace~$\cT$ of Example~\ref{ex:ex-lab}.
By considering the static tracing, the symbol $f$ within the term $g(f(a))$ does not leave a trace to the term $g(b)$ since $f$ belongs to redex pattern of $r_1$. 
Contrarily, $f$ dynamically traces to $b$. 
Finally, 
in both cases the symbol $a$ is discarded without leaving a trace in the rewrite step.
\end{example}

As for the dynamic tracing relation, the symbols can be partitioned into {\em needed} and {\em non-needed}. 
A symbol is called {\em needed} if it leaves a trace in the considered rewrite sequence.
For instance, in the previous example, $f$ is a {\em needed} symbol.
Instead $a$, which belongs to substitution $\sigma = \{x/a\}$, is a {\em non-needed} symbol.
Given an execution trace, the set of needed symbols in a term of the trace forms a prefix which is also called {\em needed} prefix.

Typically, tracing is implemented by means of labeling, i.e., the objects are labeled to be traced along the execution trace. 
For instances, let us consider Klop labeling for a rewrite step  $t \raw s$. 
A symbol in $t$ traces to a symbol in $s$, if and only if the label of the former is a sublabel of the label of the latter.
Note that this tracing relation is independent of the chosen tracing, while it is strictly tied to the labeling strategy.

Labeling and tracing relations in term rewriting systems have been studied in~\cite{Terese03book}.
In order to study the orthogonality of execution traces, \cite{Terese03book}~investigates the equivalence of labeling and tracing along with other characterizations such as permutation, standardization, and projection. As far as we know, the use of labeling and tracing for model checking and debugging purposes has not been previously discussed in the related literature.

\section{Backward Trace Slicing for Elementary Rewrite Theories}\label{sec:labeling_and_slicing}

In this section, we formalize a backward trace slicing technique for {\em elementary rewrite theories} that is based on a term labeling procedure that is inspired by~\cite{BKV00}. 
Since equations in $\Delta$ are treated as rewrite rules that are used to simplify terms, %
our formulation for the  trace slicing technique %
 is purely based on standard rewriting.

\subsection{Labeling procedure for rewrite theories}\label{sec:labeling}

Let us define a labeling procedure for rules 
similar to~\cite{BKV00} that allows us to trace %
symbols involved in a rewrite step.
First, we provide the notion of labeling for terms, and   then we show how it can be naturally lifted to rules and rewrite steps.

Consider a set  $\cA$ of \emph{atomic labels}, which are denoted by Greek letters $\alpha, \beta,\ldots$.
\emph{Composite labels} (or simply \emph{labels}) are defined as finite sets of elements of $\cA$.   %
By abuse, 
we write the label $\alpha\beta\gamma$ %
as a compact denotation for the set $\{\alpha, \beta,\gamma\}$.

A \emph{labeling} for a term $t\in\Tc_{\Sigma\cup\{\Box\}}(\Variables)$ is a  map  $L$ that assigns a label to (the symbol occurring at) each %
position $w$ of $t$, provided that $root(t_{|w}) \neq\Box$. 
If $t$ is a term, then $t^L$ denotes the labeled version of $t$. Note that, in the case when $t$ is a context, occurrences of symbol $\Box$ appearing in the labeled version of $t$ are not labeled.
The \emph{codomain} of a labeling $L$ is denoted by $\mathit{Cod}(L)=\{l\mid (w\mapsto l)\in L\}$.
An {\em initial labeling} for the term $t$ 
is a labeling for $t$ that assigns  distinct %
fresh atomic labels to each 
position of the term.
For example, given $t=f(g(a,a),\Box)$, then $t^L = f^\alpha(g^\beta(a^\gamma,a^\delta),\Box)$ is the labeled version of $t$ via the  initial labeling 
${L=}\{{\Lambda\mapsto\alpha},$ ${1\mapsto\beta}$, ${1.1\mapsto\gamma},$ ${1.2\mapsto\delta}\}$.
This notion extends to rules and rewrite steps in a natural way as shown below.

\subsubsection{Labeling of Rules.}

The labeling of a rewriting rule is formalized as follows:

\begin{definition}\label{def:ruleLabel} (rule labeling) \cite{BKV00}
Given a rule $r : \lambda \raw \rho$, a labeling  $L_{r}$ for $r$ is defined by means of the following procedure.

\begin{itemize}
\item[$r_1.$] 

The redex pattern $\lambda^\square$ %
is labeled by means of an initial labeling $L$.

\item[$r_2.$] A new label $l$ %
is formed by joining  
all the labels that occur  in the labeled redex pattern $\lambda^\square$
(say in alphabetical order) of the rule $r$. Label $l$ is then associated with each
position $w$ of the contractum pattern $\rho^\square$, %
provided that $root(\rho^\square_{|w})\neq \Box$.
\end{itemize}

\end{definition}

The labeled version of $r$ w.r.t. $L_{r}$ is denoted by $r^{L_r}$.
Note that the labeling procedure shown in Definition \ref{def:ruleLabel} does not assign labels to variables  but only to the function symbols occurring in the rule. 

\subsubsection{Labeling of Rewrite Steps.}
Before giving the definition of labeling for a rewrite step, we need to formalize the auxiliary notion of substitution labeling.

\begin{definition}(substitution labeling) \label{def:susbtL}
Let $\sigma=\{x_1/t_1,\ldots,x_n/t_n\}$ be a substitution.
A labeling $L_{\sigma}$ for the substitution $\sigma$ is defined by a set of initial labelings
$L_\sigma=\{L_{x_1/t_1},\ldots,L_{x_n/t_n} \}$ 
such that 
(i) for each binding $(x_i/t_i)$ in the substitution $\sigma$, $t_i$ is labeled using the corresponding initial labeling $L_{x_i/t_i}$, and 
(ii) the sets $\mathit{Cod}(L_{x_1/t_1}),\ldots,\mathit{Cod}(L_{x_n/t_n})$ are pairwise disjoint.
\end{definition}

By using Definition~\ref{def:susbtL}, we can formulate a labeling procedure for rewrite steps as follows.

\begin{definition}\label{def:labelled-step} (rewrite step labeling)  
Let $r:\lambda\rightarrow \rho$ be a rule, and  
$\mu : t\stackrel{r,\sigma}{\rightarrow}s$
be a rewrite step using $r$ such that  $t = C[\lambda\sigma]_q$ and $s= C[\rho \sigma]_q$, for a context $C$ and position $q$.
Let $\sigma=\{x_1/t_1,\ldots,x_n/t_n\}$. 
Let $L_r$ be a labeling for the rule $r$, $L_C$ be an initial labeling for the context $C$, 
and $L_\sigma=\{L_{x_1/t_1},\ldots,L_{x_n/t_n}\}$ be a labeling for the substitution $\sigma$ such that the sets 
$\mathit{Cod}(L_C),\mathit{Cod}(L_r)$, and $\mathit{Cod}(\sigma)$ are pairwise disjoint, where $\mathit{Cod}(\sigma)=\bigcup_{i=1}^n
\mathit{Cod}(L_{x_i/t_i}).$

The {\em rewrite step} labeling $L_{\mu}$ for $\mu$  is defined by successively applying the following steps:

\begin{itemize}

\item[$s_1.$]  %
First, positions of $t$ or $s$ that belong to the context $C$ are labeled by using the initial labeling $L_C$. 

\item[$s_2.$]  %
Then positions of $t_{|q}$ (resp. $s_{|q}$) that correspond to  the redex pattern (resp. contractum pattern) of the rule $r$ rooted at the position $q$ are labeled according to the labeling~$L_r$. 

\item[$s_3.$]  
Finally, for each term $t_j$, $j=\{1,\ldots,n\}$, which has been introduced in $t$ or $s$ via the binding $x_j/t_j\in\sigma$, 
with  $x_j\in Var(\lambda)$, $t_j$ is labeled using  the corresponding labeling $L_{x_j/t_j}\in L_\sigma$
\end{itemize}
\end{definition}

The labeled version of a rewrite step $\mu$ w.r.t. $L_{\mu}$ is denoted by $\mu^{L_\mu}$.
Let us illustrate these definitions by means of a rather intuitive example.

\begin{example}\label{ex:label-whole}
Consider the rule $r : f(g(x,y),a) ) \rightarrow d(s(y),y)$.
The labeled version of rule $r$ using the initial labeling $L=\{(\Lambda\mapsto\alpha,1\mapsto\beta,2\mapsto\gamma\}$ is as follows:
{\small $$ f^\alpha(g^\beta(x,y),a^\gamma) \rightarrow  d^{\alpha\beta\gamma}(s^{\alpha\beta\gamma}(y),y)$$}
Consider a rewrite step  $\mu : C[\lambda\sigma] \stackrel{r}{\rightarrow} C[\rho\sigma] $  using $r$, where $C[\lambda\sigma] = d(f(g(a,h(b)),a),a)$, $C[\rho\sigma]= d(d(s(h(b)),h(b)),a)$, and
${\sigma = \{x/a, y/h(b)\}}$. 
Let  $L_{C}=\{\Lambda\mapsto\delta,~ 2\mapsto\epsilon\}$, 
$L_{x/a} = \{\Lambda\mapsto\zeta\}$, and
$L_{y/h(b)}= \{\Lambda\mapsto\eta, 1\mapsto\theta \}$ be the labelings for $C$ and the bindings in $\sigma$, respectively.
Then, the corresponding labeled rewrite step $\mu^L$ %
is  as follows
{\small $$
\mu^L : d^\delta(f^\alpha(g^\beta(a^\zeta, h^\eta(b^\theta)), a^\gamma), a^\epsilon)
\raw 
d^\delta(d^{\alpha\beta\gamma}(s^{\alpha\beta\gamma}(h^\eta(b^\theta)), h^\eta(b^\theta)),a^\epsilon)
$$}
\end{example}

\subsection{Backward Tracing Relation}

Given a rewrite step $\mu : t \stackrel{r}{\rightarrow} s$ and the labeling process defined in the previous section, the {\em backward tracing relation} computes the set of positions in %
$t$ that are origin for a position $w$ in $s$. %
Formally.

\begin{definition}\label{def:tracing}
(origin positions)
Let $\mu : t \xrightarrow{r} s$ be a rewrite step and 
$L$ be a labeling for $\mu$ where $L_t$ (resp.\ $L_s$) is the labeling of $t$ (resp.\ $s$).
Given a position  $w$ of $s$, the set of origin positions of $w$ in $t$  
w.r.t.\  $\mu$ and $L$ (in symbols, $\lhd_\mu^L w$) is  defined as follows:
{\small $$
\begin{array}{r}
\lhd_\mu^L w= \{ v \in \pos{}(t) \mid 
 \exists p \in \pos(s), (v \mapsto l_v ) \in L_t, (p \mapsto l_p) \in L_s \mbox{ s.t. } p\leq w %
 
 \mbox{ and } l_v \subseteq l_p \}
\end{array}
$$}
\end{definition}
\medskip

Note that Definition~\ref{def:tracing} considers all positions of $s$ in the path from its root to $w$ for computing the origin positions of $w$.
Roughly speaking,  a  position $v$ in $t$  is an origin of $w$, if
the label of the symbol that occurs in $t^L$ at position  $v$ is contained in the label of a symbol that occurs   in $s^L$ in the path from its 
root to the position $w$. 

\begin{example}\label{ex:back-tracing-relation}
Consider  again the rewrite step $\mu^L : t^L {\rightarrow}s^L$
of Example~\ref{ex:label-whole}, and let $w$ be the position $1.2$ %
of $s^L$.
The set of labeled symbols  occurring  in $s^L$ in the path from its
root to position $w$ is the set
$ z=\{ h^\eta, d^{\alpha\beta\gamma}, d^\delta \}$.
Now, the labeled symbols occurring  in $t^L$ whose label is contained in the label of one element of $\tt z$ is the set
$ \{ h^\eta, f^\alpha, g^\beta, a^\gamma, d^\delta \}$.
By Definition~\ref{def:tracing}, 
 the set of origin positions of $w$ in  $\mu^L$ is 
$  \lhd_\mu^L w =  \{1.1.2,~1,~1.1,~1.2,~\Lambda \}$.
\end{example}

\subsection{The Backward Trace Slicing Algorithm}\label{sec:slicing}

First, let us formalize the slicing criterion, which %
basically represents the information %
we want to trace back across the execution trace in order to find out the ``origins'' of the data we observe.
Given a term $t$, we denote by $\cO_t$ the set of {\em observed} positions  of $t$. %

\begin{definition}\label{def:slicingCriterion}
{(slicing criterion)}
Given a rewrite theory  $\cR=(\Sigma,\Delta,R)$ and  
 an execution trace %
${\cT : s \rightarrow^* t}$ in $\cR$,
a slicing criterion for  ${\cT}$ is any set $\cO_{t}$
of positions of the term $t$. 
\end{definition}

In the following, we show how backward trace slicing can be performed by exploiting the backward tracing relation $\lhd_\mu^L$
that was introduced in Definition~\ref{def:tracing}.
Informally, given a slicing criterion~$\cO_{t_{n}}$ %
 for $\cT:t_0\raw t_2\raw\ldots \raw t_n$,  at each rewrite step  $t_{i-1}\raw t_{i}$, $i=1,\ldots, n$,
our technique inductively computes the backward tracing relation between the relevant positions of   $t_i$ and those in $t_{i-1}$. The algorithm proceeds backwards, from the final term $t_n$ to the initial term $t_0$, and recursively generates at step $i$ the corresponding set of relevant positions,  $P_{t_{n-i}}$. %
Finally, by means of a %
removal function, 
a simplified trace is obtained where each $t_j$  is replaced by the corresponding {\em term slice} that contains only the relevant information w.r.t. $P_{t_{j}}$.

\medskip
\begin{definition}\label{def:rlvSym}
{(sequence of relevant position sets)} 
Let $\cR=(\Sigma,\Delta,R)$ be a rewrite theory, and $\cT: t_0 \stackrel{r_1}{\rightarrow} t_1  \ldots  \stackrel{r_n}{\rightarrow} t_n$
 be an execution trace in $\cR$.
Let $L_i$ be the labeling for the rewrite step $t_i \rightarrow t_{i+1}$ with $0 \leq i <  n$.
The sequence of  relevant position sets  in $\cT$ w.r.t. the slicing criterion  %
$\cO_{t_{n}}$  is defined as follows:
$$\begin{array}{l}
relevant\_positions(\cT, \cO_{t_{n}}) = [P_0, \ldots, P_n] \\ %
\mbox{where }  
\begin{cases}
P_n = \cO_{t_{n}} \\
P_j = \bigcup_{p\in P_{j+1}}  \lhd^{L_j}_{(t_{j} \overarrow{} ~t_{j+1})} p,
\mbox{ with }  0 \leq j <  n  
\end{cases}
\end{array}
$$

\end{definition}
\medskip

Now, it is straightforward to  formalize a procedure that obtains a term slice from each term $t$ in $\cT$ and the corresponding set of relevant positions of $t$.
We introduce the fresh symbol $\bullet \not\in\Sigma$ to 
replace any information in the term  that is not relevant, 
hence does not affect the observed criterion. 

\begin{definition}\label{def:termSlice}
{(term slice)} 
Let $t \in \Tc_{\Sigma}$ be a term and  
$P$ be a set of positions of $t$.
A \emph{term slice}  of $t$ with respect  to %
 $P$    is defined as follows: 
$$slice(t,P) =  sl\_rec(t,P,\Lambda), ~where$$
$$
sl\_rec(t,P,p) = \left \{ \begin{array}{l} 
   f(sl\_rec(t_1, P, p.1), \ldots, sl\_rec(t_n,P,p.n))  \\
    \hspace*{.7cm}  
    \mbox{ if } t= f(t_1,\ldots,t_n) \mbox{ and there exists } w \mbox{ s.t. } (p.w) \in P  \\ 
  \bullet  \hspace*{.5cm} \mbox{ otherwise}
\end{array}\right. 
$$

\end{definition}

In the following, we use the notation $t^\bullet$ %
to denote a term slice of the term~$t$.
Roughly speaking, the symbol $\bullet$ can be thought of as a variable,
so that any term $t' \in \tau(\Sigma)$ can be considered as a possible
concretization of $t^\bullet$ if it is an ``instance" of $[t^\bullet]$,
where $[t^\bullet]$ is the term that is obtained by replacing all occurrences
of $\bullet$ in $t^\bullet$ with fresh variables.

\begin{definition}\label{def:app} (term slice concretization)
Given   $t' \in \Tc_{\Sigma}$ %
and a term slice $t^\bullet$, we define $t^\bullet \propto t'$ if $[t^\bullet]$   is (syntactically)  more general than $t'$ (i.e.,  $[t^\bullet]\sigma=t'$, for some substitution $\sigma$). We also say that
$t'$ is a concretization of  $t^\bullet$.
\end{definition}

Figure~\ref{fig:slice} illustrates the notions of term slice and  term slice concretization for a given term $t$ w.r.t.\ the set of positions $\{1.1.2,1.2\}$.

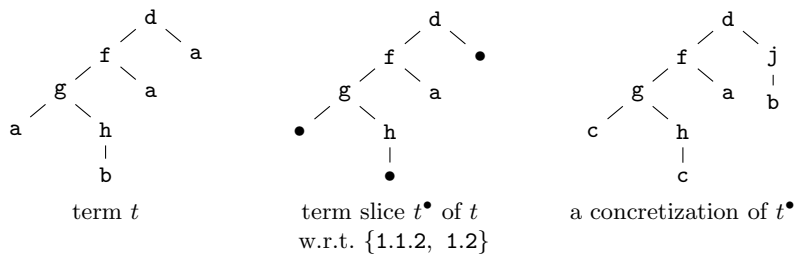
\begin{figure}[t!]
 \centering
$$
\begin{array}{c@{\hspace{1cm}}c@{\hspace{1cm}}c}
\begin{tikzpicture}  [level distance=.5cm]
\node {$\tt d$} [sibling distance=1.2cm]
child{ node {$\tt f$} [sibling distance=1.2cm]
  child{ node {$\tt g$} [sibling distance=1.2cm]
     child{ node {$\tt a$} }
     child{ node {$\tt h$} [level distance=.6cm] child{ node {$\tt b$} }
     }
  }
  child{ node {$\tt a$ } [sibling distance=1.2cm]
  }
}
child{ node {$\tt a$} }
;
\end{tikzpicture} 
 & 
 \begin{tikzpicture}  [level distance=.5cm]
\node {$\tt d$} [sibling distance=1.2cm]
child{ node {$\tt f$} [sibling distance=1.2cm]
  child{ node {$\tt g$} [sibling distance=1.2cm]
     child{ node {$\bullet$} }
     child{ node {$\tt h$} [level distance=.6cm] child{ node {$\bullet$} }
     }
  }
  child{ node {$\tt a$ } [sibling distance=1.2cm]
  }
}
child{ node {$\bullet$} }
;
\end{tikzpicture} 
 & 
 \begin{tikzpicture}  [level distance=.5cm]
\node {$\tt d$} [sibling distance=1.2cm]
child{ node {$\tt f$} [sibling distance=1.2cm]
  child{ node {$\tt g$} [sibling distance=1.2cm]
     child{ node {$\tt c$} }
     child{ node {$\tt h$} [level distance=.6cm] child{ node {$\tt c$} }
     }
  }
  child{ node {$\tt a$ } [sibling distance=1.2cm]
  }
}
child{ node {$\tt j$} [level distance=.6cm] [sibling distance=1.2cm]
     child{ node {$\tt  b$} }
}
;
\end{tikzpicture} 
 \\
\mbox{term } t & \mbox{term slice } t^\bullet \mbox{ of } t & \mbox{a concretization of } t^\bullet \\
 &  \mbox{ w.r.t.\ } \tt \{1.1.2,~ 1.2\} &  
\end{array}
$$

\caption{A term slice and a possible concretization.}\label{fig:slice}
\end{figure}

Let us  define a {\em sliced rewrite step} between two  term slices as follows.

\begin{definition}\label{def:sliced_rewrite}
(sliced rewrite step)
Let $\cR=(\Sigma, \Delta, R)$ be a rewrite theory and $r$ a rule of $\cR$.
The term slice $s^\bullet$ rewrites to the term slice $t^\bullet$ via $r$ (in symbols, $s^\bullet \overarrow{r} t^\bullet$) if there exist two terms $s$ and $t$  
such that $s^\bullet$ is a term slice of $s$, $t^\bullet$ is a term slice of $t$, and $s \overarrow{r} t$. 
\end{definition}
\medskip

Finally, using Definition \ref{def:sliced_rewrite}, backward %
trace slicing   %
is formalized  as follows.

\medskip
\begin{definition}\label{def:bak-dyn-tech}
{(backward trace slicing)}
Let $\cR=(\Sigma,\Delta,R)$ be a rewrite theory, and  $\cT: t_0 \stackrel{r_1}{\rightarrow} t_1  \ldots  \stackrel{r_n}{\rightarrow} t_n$ be an execution trace in $\cR$.
Let $\cO_{t_n}$ be a slicing criterion for $\cT$, and 
let $[P_0, \ldots, P_n]$ be the sequence of the relevant position sets of $\cT$ w.r.t.\  $\cO_{t_n}$.
A trace slice $\cT^\bullet$ of $\cT$ w.r.t.\ $\cO_{t_n}$ %
is defined as the sliced rewrite sequence of term slices $t_i^\bullet = slice(t_i, P_i)$ which is obtained by 
gluing together the sliced rewrite steps in the set 
$${\cal K}^\bullet = \{ t_{k-1}^\bullet \stackrel{r_k}{\rightarrow} t_k^\bullet \mid 0 < k \leq n ~\wedge~ t_{k-1}^\bullet \neq t_k^\bullet\}.$$
\end{definition}

Note that in Definition~\ref{def:bak-dyn-tech}, the  sliced rewrite steps that %
do not affect the relevant positions %
(i.e., $t_{k-1}^\bullet \stackrel{r_k}{\rightarrow} t_k^\bullet$ with  $t_{k-1}^\bullet=t_{k}^\bullet$)
are discarded, which further reduces the size of the trace.

A desirable property of a slicing technique is to ensure that, for any 
concretization of the term slice $t^\bullet_0$, 
the trace slice $\cT^\bullet$ can be reproduced.
This property  ensures that the rules involved in  $\cT^\bullet$ can be applied again to every concrete trace $\cT'$ that we can derive by instantiating all the variables in $[t^\bullet_0]$ with arbitrary terms.

\begin{theorem}\label{prop:reproduded} (soundness) 
Let $\cR$ be an elementary rewrite theory. 
Let $\cT$ be an  execution trace in the rewrite theory $\cR$, and let $\cO$ be a slicing criterion for $\cT$.
Let $\cT^\bullet : t_0^\bullet \stackrel{r_1}{\rightarrow} t_1^\bullet  \ldots \stackrel{r_n}{\rightarrow} t_n^\bullet $ be the corresponding  trace slice w.r.t.\ $\cO$. %
Then, for any concretization  
$t_0'$ of $t_0^\bullet$, 
it holds that $\cT':t_0' \stackrel{r_1}{\rightarrow} t_1' \ldots \stackrel{r_n}{\rightarrow} t_n'$
is an  execution trace in $\cR$, and   $t_i^\bullet \propto t_i'$, for $i=1,\ldots,n$.
\end{theorem}

The proof of Theorem~\ref{prop:reproduded} relies on the fact that  redex patterns are preserved  by backward trace slicing. Therefore,  for $i=1,\ldots,n$, the rule $r_i$ can be applied to any concretization $t_{i-1}'$ of term $t_{i-1}^\bullet$ since the redex pattern of $r_i$ does appear  in $t_{i-1}^\bullet$, and hence in $t_{i-1}'$. 
A detailed proof of Theorem~\ref{prop:reproduded} can be found in Appendix~\ref{app:proof}.

Note that our basic framework enjoys neededness of the extracted information (in the sense of~\cite{Terese03book}),
since the  information captured by every sliced rewrite step in a trace slice is all and only the information that is needed to produce the data of interest in the reduced term.

\section{Backward Trace Slicing for Extended Rewrite Theories}\label{sec:slicing-m-e}

In this section, we consider an extension of our basic slicing methodology that allows us to deal with extended rewrite theories $\cR=(\Sigma,E,R)$ where the equational theory $(\Sigma, E)$ may contain associativity and commutativity axioms, and $R$ may contain collapsing as well as nonleft-linear rules. 
Moreover, we also consider the built-in operators, which are not equipped with an explicit functional definition (e.g., Maude arithmetical operators). %
It is worth noting that all the proposed extensions are restricted to the labeling procedure of Section~\ref{sec:labeling}, keeping the backbone of our slicing technique unchanged.  

\subsection{Dealing with collapsing and nonleft-linear rules}

\noindent\textbf{Collapsing Rules.}
The main difficulty with collapsing rules is that they have a trivial contractum pattern, which consists in the empty context~$\Box$; hence,
  it is not possible to propagate labels 
from the left-hand side of the rule to its right-hand side. This makes the rule labeling procedure of Definition~\ref{def:ruleLabel} completely unproductive
for trace slicing.

In order to overcome this problem, we keep track of the labels in the left-hand side of the  collapsing rule 
$r$, whenever a rewrite step involving $r$ takes place. This amounts to extending the labeling procedure of Definition \ref{def:labelled-step} as follows.

\begin{definition}\label{def:collapsext}
(rewrite step labeling for collapsing rules)
Let $\mu : t\stackrel{r,\sigma}{\rightarrow}s$ be a rewrite step s.t.\ $\sigma=\{x_1/t_1,\ldots,x_n/t_n\}$, where $r: \lambda\rightarrow x_i$ is a collapsing rule. 
Let $L_r$ be a labeling for the rule $r$.
In order to label the step $\mu$, we extend the labeling procedure formalized in  Definition \ref{def:labelled-step} as follows:
\begin{itemize}
\item[$s_4$.] Let $t_i$ be the term  introduced in $s$  via the binding $x_i/t_i\in\sigma$, for some $i \in \{1,\ldots,n\}$. Then,
the label $l_i$ of the root symbol of $t_{i}$ in $s$ is replaced by a new  composite label  $l_cl_i$, where $l_c$ is formed   by  joining all the labels appearing in the redex pattern of $r^{L_{r}}$.
\end{itemize}
\end{definition}

\noindent\textbf{Nonleft-linear Rules.}
The trace slicing technique we described so far %
does not work for nonleft-linear TRS.
Consider the rule: $r :  f(x,y,x) \raw g(x,y)$
and the one-step trace  $\cT:f(a,b,a)\raw g(a,b)$.
If we are interested in tracing back the symbol $g$ 
that occurs in the final state $g(a,b)$, 
we would get the following trace slice 
$\cT^\bullet : f(\bullet, \bullet, \bullet) \rightarrow g(\bullet,\bullet)$. 
However, $f(a,b,b)$ is a concretization of  $f(\bullet, \bullet,\bullet)$ 
that cannot be rewritten by using $r$. 
In the following, we  augment  Definition~\ref{def:collapsext} in order to also
deal with nonleft-linear rules.

\begin{definition}\label{def:nllext}
(rewrite step labeling for nonleft-linear rules)
Let $\mu : t\stackrel{r,\sigma}{\rightarrow}s$ be a rewrite step s.t.\ $\sigma=\{x_1/t_1,..,x_n/t_n\}$, where $r$ is a nonleft-linear rule.
Let $L_\sigma=\{L_{x_1/t_1},.., L_{x_n/t_n}\}$ be a labeling for the substitution $\sigma$.
In order to label the step $\mu$, we further extend the labeling  procedure 
formalized in  Definition \ref{def:collapsext}  as follows:
\begin{itemize}
\item 
[$s_5$.] For each variable $x_j$ that occurs more than once in the left-hand side of the rule $r$, the following steps must be followed:
 \begin{itemize} \item we form a new label $l_{x_j}$ by joining all the labels in  $\mathit{Cod}(L_{x_j/t})$ where $L_{x_j/t}\in L_\sigma$;
\item 
let $l_s$ be the label of the root symbol of $s$.
Then, $l_s$ is replaced by a new  composite label  $l_{x_j}l_s$.

\end{itemize}
\end{itemize}
\end{definition}
Note that, whenever a rewrite step $\mu$ involves the application of  a rule that is both collapsing and non left-linear, the labeling for $\mu$ is obtained by sequentially applying step $s_4$ of Definition \ref{def:collapsext} and step $s_5$ of Definition \ref{def:nllext} (over the labeled rewrite step resulting from $s_4$).

\begin{example}
Consider the labeled, collapsing and nonleft-linear rule 
${f^\beta (x,y,x) \rightarrow y}$ 
together with the rewrite step  
${\mu: h( f(a, b, a), b )  \rightarrow h( b, b)}$, and matching substitution 
$\sigma = \{x/a,y/b\}$.
Let $L_{h(\Box,b)} = \{ \Lambda\mapsto\alpha, 2\mapsto\epsilon \}$ be the labeling for the context $h(\Box,b)$.
Then, for the labeling 
$L_\sigma=\{ L_{x/a},~ L_{y/b} \}$, with
$L_{x/a} = \{ \Lambda\mapsto\gamma \}$ and 
$L_{y/b} = \{ \Lambda\mapsto\delta \}$,
the labeled version of $\mu$ is  %
$h^\alpha ( f^\beta(a^\gamma, b^\delta, a^\gamma), b^\epsilon ) \rightarrow 
h^\alpha( b^{\beta\gamma\delta}, b^\epsilon )$. %
Finally, by considering the criterion $\{1 \}$, we can safely trace back the symbol $b$ of the sliced final state $h( b, \bullet)$
and obtain the following trace slice
$$h(f(g(a),b,g(a)), \bullet) \rightarrow h( b, \bullet) .$$ 
\end{example}

\subsection{Built-in Operators}\label{sec:orden-total}

In practical implementations of RWL (e.g., Maude \cite{maude-book}), several commonly used operators  are pre-defined (e.g., arithmetic operators, if-then-else constructs), which do not have an explicit specification.
To overcome this limitation, we further extend our labeling process in order to deal with built-in operators.
\begin{definition}\label{def:nbuilt-in}
(rewrite step labeling  for built-in operators)
For the case of a rewrite step  $\mu: C[op(t_1,\ldots,t_n)] \raw C[t']$  involving a call to a built-in, $n$-ary operator $op$,   
 we extend Definition \ref{def:nllext} by introducing the following additional case:
\begin{itemize}
\item [$s_6$.] 
Given   an initial labeling $L_{op}$  for the term  $op(t_1,\ldots,t_n)$,
\begin{itemize}                
\item 
each symbol occurrence in $t'$ is labeled with  a new label %
that is formed by joining the labels of all the (labeled) arguments $t_1,\ldots,t_n$ of $op$;
\item the remaining symbol occurrences of $C[t']$ that are not considered in the previous step inherit all the labels appearing in $C[op(t_1,\ldots,t_n)] $.%
\end{itemize}
\end{itemize}
\end{definition}
For example, by applying Definition \ref{def:nbuilt-in}, the addition of two natural numbers implemented through the built-in operator $+$ might be labeled as  ${+^\alpha(7^\beta,8^\gamma) \rightarrow 15^{\beta\gamma}}$.%

\subsection{Associative-Commutative Axioms}
Let us finally consider an extended rewrite theory $\cR=(\Sigma,\Delta\cup B,R)$, where $B$ is a set of associativity (A) and commutativity (C) axioms that hold for some function symbols in $\Sigma$.   
Now, since $B$ only contains associativity/commutativity (AC) axioms, terms can be  represented  by means of a single representative of their AC congruence class, called {\em AC canonical form}~\cite{Eke03}.
This representative is obtained by replacing nested occurrences of the same AC operator by a flattened argument list under a variadic symbol, whose elements are sorted by means of some linear ordering~\footnote{Specifically, Maude uses the lexicographic order of symbols.}.
The inverse process to the flat transformation is the unflat transformation, which is nondeterministic (in the sense that it generates all the unflattended terms that are equivalent (modulo AC) to the  flattened term)~\footnote{These two processes are typically hidden inside the $B$-matching algorithms  that are used to implement rewriting modulo $B$. See~\cite{maude-book} (Section~$4.8$) for an in-depth discussion on matching and simplification modulo AC in Maude.}.

For example, consider a binary AC operator $f$  together with the standard lexicographic ordering over symbols. 
Given the $B$-equivalence $f(b,f(f(b,a),c)) =_B f(f(b,c),f(a,b))$, we can represent it by using the ``internal sequence'' $f(b,f(f(b,a),c)) \raw^*_{\mathit{flat}_B} f(a,b,b,c) \raw^*_{\mathit{unflat_B}} f(f(b,c),f(a,b))$, where the first one corresponds to the {\em flattening} transformation sequence that obtains the AC canonical form, while the second one corresponds to the inverse, unflattening~one.

The key idea for extending our labeling procedure in order to cope with $B$-equivalence $=_B$ is to exploit the flat/unflat transformations
mentioned above. 
Without loss of generality,
we assume that flat/unflat transformations are stable w.r.t.\ the lexicographic ordering over positions $\sqsubseteq$\footnote{The lexicographic ordering $\sqsubseteq$ is defined as follows: 
$\Lambda \sqsubseteq w$ for every position $w$, and given the positions $w_1 = i.w_1'$ and $w_2= j.w_2'$, 
$w_1 \sqsubseteq w_2$ iff $i < j$ or ($i = j$ and $w_1' \sqsubseteq w_2'$).
Obviously, in a practical implementation of our technique, the
considered ordering among the terms should be chosen to
agree with the ordering considered by flat/unflat
transformations in the RWL  infrastructure.}. %
This assumption %
allows us to trace back arguments of commutative operators, since multiple occurrences of the same symbol can be precisely identified.

\begin{definition}\label{def:nAC}
(AC Labeling.)
Let $f$ be an associative-commutative operator and  $B$ be the AC  axioms for $f$.
Consider the
$B$-equivalence $t_1 =_B t_2$ and the corresponding (internal) flat/unflat transformation $\cT: t_1 \raw^*_{\mathit{flat}_B} s \raw^*_{\mathit{unflat}_B} t_2$.
Let $L$ be an initial labeling for $t_1$. 
The labeling procedure for $t_1 =_{B} t_2$ is as follows.

\begin{enumerate}
\item \label{flat} (flattening) For each flattening transformation step
$t_{|v} \raw_{\mathit{flat}_B} t'_{|v}$ in $\cT$ for the symbol $f$,
a new label $l_f$ is formed by joining all the labels attached to the symbol $f$ in any position $w$ of
$t_{}^L$ s.t.\   %
${w = v}$ or  
$w\geq v$, and
every symbol on the path from $v$ to $w$ is $f$;
then, label $l_f$ is attached to the root symbol of~$t'_{|v}$. 

\item \label{unflat} (unflattening)  For each unflattening transformation step $t_{|v} \raw_{\mathit{unflat}_B} t'_{|v}$ in $\cT$ for the symbol $f$,
the label of the symbol $f$ in the position $v$ of $t^L$ is attached to the symbol $f$ in any position $w$ of $t'$
such that ${w = v}$ or
$w\geq v$, and every symbol on the path from $v$ to $w$ is $f$.

\item
The remaining symbol occurrences  in $t'$ that are not considered in cases \ref{flat} or \ref{unflat} above
inherit the label of the corresponding symbol occurrence in $t$. %

\end{enumerate}

\end{definition}

\begin{example}
Consider the transformation sequence %
$$\small \begin{array}{c}
f(b,f(b,f(a,c))) \raw^*_{\mathit{flat}_B} f(a,b,b,c) \raw^*_{\mathit{unflat}_B}    f(f(b,c),f(a,b))
   \end{array}
$$
by using Definition \ref{def:nAC}, the associated transformation sequence can be labeled as follows:
$$\small
\begin{array}{ll}
 f^\alpha(b^\beta,f^\gamma(b^\delta,f^\epsilon(a^\zeta,c^\eta))) \raw^*_{\mathit{flat}_B}& f^{\alpha\gamma\epsilon}(a^\zeta,b^\beta,b^\delta,c^\eta) \raw^*_{\mathit{unflat}_B} 
 \\ & \hspace*{1cm} 
 f^{\alpha\gamma\epsilon}(f^{\alpha\gamma\epsilon}(b^\beta,c^\eta),f^{\alpha\gamma\epsilon}(a^\zeta,b^\delta))
\end{array}
$$%
Note that the original order between the two occurrences of the constant $b$ is not changed by the flat/unflat transformations.
For example, in the first term, $b^\beta$ is in position~$1$ and $b^\delta$ is in position $2.1$ with $1 \sqsubseteq 2.1$, 
whereas, in the last term, $b^\beta$ is in position~$1.1$ and $b^\delta$ is in position $2.2$ with $1.1 \sqsubseteq 2.2$.
\end{example}

Finally, note that the methodology described in this section can be easily extended to deal with other equational attributes, e.g., identity (U), by explicitly encoding the internal transformations performed %
via suitable rewrite rules.

\subsection{Extended Soundness}

Soundness of the backward trace slicing algorithm for the extended rewrite theories is established by the following
theorem which properly extends Theorem~\ref{prop:reproduded}. The proof of such an extension can be found in Appendix~\ref{app:proof}. 

\begin{theorem}\label{prop:reproduded2} (extended soundness) 
Let $\cR = (\Sigma, E, R)$ be an extended rewrite theory. 
Let $\cT$ be an  execution trace in the rewrite theory $\cR$, and let $\cO$ be a slicing criterion for $\cT$.
Let $\cT^\bullet : t_0^\bullet \stackrel{r_1}{\rightarrow} t_1^\bullet  \ldots \stackrel{r_n}{\rightarrow} t_n^\bullet $ be the corresponding  trace slice w.r.t.\ $\cO$. 
Then, for any concretization  
$t_0'$ of $t_0^\bullet$, 
it holds that $\cT':t_0' \stackrel{r_1}{\rightarrow} t_1' \ldots \stackrel{r_n}{\rightarrow} t_n'$
is an  execution trace in $\cR$, and   $t_i^\bullet \propto t_i'$, for $i=1,\ldots,n$.
\end{theorem}

\section{Experimental Evaluation}\label{sec:exp}

We have developed a prototype implementation of our slicing methodology that is publicly available at \url{http://www.dsic.upv.es/~dromero/slicing.html}. 
The implementation is written in Maude and consists of approximately 800 lines of code.
Maude is a high-performance, reflective language that supports both equational and rewriting logic programming, which is particularly suitable for developing domain-specific applications~\cite{EMS03}.
The reflection capabilities of Maude allow metalevel computations in RWL to be handled at the object-level. 
This facility allows us to easily manipulate computation traces of Maude itself and eliminate the irrelevant contents by implementing the backward slicing procedures that we have defined in this paper.
Using reflection to implement the slicing tool has one important
additional advantage, namely, the ability to quickly integrate the tool within the Maude formal tool environment \cite{ClavelDHLMO07}, 
which is also developed using reflection.

We developed the operator {\tt slice} that implements the slicing process.
This operator is invoked as follows: 
$$\tt slice(\langle moduleName \rangle, \langle initialState \rangle, \langle endState \rangle, \langle criterion \rangle)$$
where $\tt moduleName$ is the name of the Maude module that includes the rules and the equations  to be considered in the slicing process;
$\tt initialState$ and $\tt endState$  are the initial state and the final state, respectively, of the execution trace; 
and $\tt criterion$ is the slicing criterion.
The operator works as follows.
First, by considering the rules and equation in $\tt moduleName$, 
the instrumented execution trace stemming from the initial state that leads to the final state is computed. 
Then, the slicing  procedure is executed with the instrumented computation  trace and the slicing criterion   as inputs.
Finally, a pair that contains the sliced trace and the original execution trace is delivered as outcome of the process.

In order to evaluate the usefulness of our approach, we benchmarked our prototype with several examples of Maude applications, namely:
{\em War of Souls} ({\tt WoS}), a role-playing game that is modeled as a nontrivial producer/consumer application;
{\em Fault-Tolerant Communication Protocol} ({\tt FTCP}), a Maude specification %
that models a fault-tolerant, client-server communication protocol; and 
Web-TLR, a software tool designed for model-checking real-size Web applications (e.g., Web-mailers, Electronic forums), which is based on rewriting logic.

We have tested our tool on some execution traces that were generated by the Maude applications described above by imposing different slicing criteria. For each application, we considered two execution traces that were sliced using two different criteria.
As for the $\mathtt{WoS}$ example, we have chosen criteria that allow us to backtrace both the values produced and the entities in play --- e.g., the criterion $\mathtt{WoS}.\cT_1.O_2$ isolates players' behaviors along the trace $\cT_1$.
Execution traces in the \texttt{FTCP} example represent client-server interactions. 
In this case, the chosen criteria aim at isolating a server and a client in a scenario that involves  multiple servers and clients ($\mathtt{FTCP}.\cT_2.O_1$), and tracking the response generated by a server according to a given client request ($\mathtt{FTCP}.\cT_1.O_1$).
In the last example, we have used Web-TLR to verify two LTL(R) properties of a Webmail application. The considered execution traces are much bigger for this program, and correspond to the counterexamples produced as outcome by the built-in model-checker of Web-TLR. In this case, the chosen criteria allow us to monitor the messages exchanged by the Web browsers and the Webmail server, as well as to focus our attention on the data structures of the interacting entities (e.g., browser/server sessions, server database).

Table \ref{tab:exp} summarizes the results we achieved.
For each criterion, Table~\ref{tab:exp} shows the size of the original trace and of the computed trace slice, both measures as the length of the corresponding string.
The {\em \%reduction} column shows the percentage of reduction achieved.
These results are very encouraging, and show an impressive reduction rate (up to $\sim 95\%$). %
Actually, sometimes the trace slices are small enough to be easily inspected by the user, who can restrict her attention 
to the part of the computation she wants to observe getting rid of  those data that are useless or even noisy w.r.t.\ 
the considered slicing criterion. 

\begin{table}[t!]
\begin{center}
{\small
\begin{tabular} {|c|c|c|c|c|c|} 
\hline
\multirow{2}{*}{\bf Example} &
{\bf Example} &
{\bf Original} &
{\bf Slicing} &
{\bf Sliced} &
{\bf\em \% }\\
 & {\bf trace} & {\bf trace size} & {\bf criterion} & {\bf trace size} & {\bf \em reduction}\\
\hline

\multirow{4}{*}{ $\mathtt{WoS}$ } 
& \multirow{2}{*}{ $\mathtt{WoS}.\cT_1$ } & 
\multirow{2}{*}{ $776$ } & {\scriptsize $\mathtt{WoS}.\cT_1.O_1$} & $201$ & $74.10\%$\\
\cline{4-6}
                     & & & {\scriptsize $\mathtt{WoS}.\cT_1.O_2$} & $138$ & $82.22\%$ \\
\cline{2-6}
& \multirow{2}{*}{ $\mathtt{WoS}.\cT_2$ } & 
\multirow{2}{*}{ $997$ } & {\scriptsize $\mathtt{WoS}.\cT_2.O_1$} & $404$ & $58.48\%$\\
\cline{4-6}
                     & & & {\scriptsize $\mathtt{WoS}.\cT_2.O_2$} & $174$ & $82.55\%$ \\
\hline

\multirow{4}{*}{ $\mathtt{FTCP}$ } 
& \multirow{2}{*}{ $\mathtt{FTCP}.\cT_1$ } & 
\multirow{2}{*}{ $2445$ } & {\scriptsize $\mathtt{FTCP}.\cT_1.O_1$} & $895$ & $63.39\%$\\
\cline{4-6}
                     & & & {\scriptsize $\mathtt{FTCP}.\cT_1.O_2$} & $698$ & $71.45\%$ \\
\cline{2-6}
& \multirow{2}{*}{ $\mathtt{FTCP}.\cT_2$ } & 
\multirow{2}{*}{ $2369$ } & {\scriptsize $\mathtt{FTCP}.\cT_2.O_1$} & $364$ & $84.63\%$\\
\cline{4-6}
                     & & & {\scriptsize $\mathtt{FTCP}.\cT_2.O_2$} & $707$ & $70.16\%$ \\
\hline

\multirow{4}{*}{ $\mathtt{Web\mbox{-}TLR}$ } 
& \multirow{2}{*}{ $\mathtt{Web\mbox{-}TLR}.\cT_1$ } & 
\multirow{2}{*}{ $31829$ } & {\scriptsize $\mathtt{Web\mbox{-}TLR}.\cT_1.O_1$} & $1949$ & $93.88\%$\\
\cline{4-6}
                     & & & {\scriptsize $\mathtt{Web\mbox{-}TLR}.\cT_1.O_2$} & $1598$ & $94.97\%$ \\
\cline{2-6}
& \multirow{2}{*}{ $\mathtt{Web\mbox{-}TLR}.\cT_2$ } & 
\multirow{2}{*}{ $72098$ } & {\scriptsize $\mathtt{Web\mbox{-}TLR}.\cT_2.O_1$} & $9090$ & $87.39\%$\\
\cline{4-6}
                     & & & {\scriptsize $\mathtt{Web\mbox{-}TLR}.\cT_2.O_2$} & $7119$ & $90.13\%$ \\
\hline

\end{tabular} }
\vspace{.2cm}
\caption{Summary of the reductions achieved.}\label{tab:exp}
\end{center}

\end{table}

\section{Conclusion and Related Work}\label{sec:related}

We have presented a backward trace-slicing technique for rewriting logic theories.
The key idea consists in tracing back ---through the rewrite sequence--- all the relevant symbols of the final state that we are interested in.
Preliminary experiments
demonstrate that the system works very satisfactorily on our benchmarks ---e.g., we obtained  trace slices that achieved a reduction of up to almost $95\%$ in reasonable time (max. 0.5s  on a Linux box equipped with an Intel Core 2 Duo 2.26GHz and  4Gb of RAM memory).

Tracing techniques have been extensively used in functional programming  for implementing debugging tools~\cite{CRW00}.
For instance, Hat~\cite{CRW00}
is an interactive debugging system that enables exploring a computation backwards, starting from the program output or an error message (with which the computation aborted).
Backward tracing in Hat is carried out by navigating a redex trail (that is, a graph-like data structure that records dependencies among function calls), whereas  tracing in our approach does not require the construction of any auxiliary data structure. 

Our backward tracing relation %
extends a previous tracing relation that was formalized in~\cite{BKV00} for orthogonal TRSs. 
In~\cite{BKV00}, a label is formed from atomic labels by using the operations of sequence concatenation and underlining (e.g., $a$, $b$, $ab$, $\underline{\underline{ab}c}d$, are labels), which are used to 
keep track of the rule application order. 
Collapsing rules are simply avoided %
by coding them away. This is done by replacing  each collapsing rule $\lambda\rightarrow x$ with the rule $\lambda\rightarrow \varepsilon(x)$, where $\varepsilon$ is a unary dummy symbol.
Then, in order to lift the rewrite relation to terms containing $\epsilon$ occurrences, infinitely many new extra-rules are added that are built by saturating all left-hand sides with $\varepsilon(x)$.
In contrast to~\cite{BKV00}, we use a simpler notion of labeling, where composite labels are interpreted as sets of atomic labels, and in the case
of collapsing as well as nonleft-linear rules we label the rewrite steps themselves so that we can
deal with these rules in an effective way.

The work that is most closely related to ours is~\cite{FT94}, which
formalizes a notion of dynamic dependence among symbols by means of contexts  %
and studies its application to program slicing of TRSs that may include collapsing as well as nonleft-linear rules.
Both the {\em creating} and the {\em created} contexts associated with a reduction (i.e., the minimal subcontext that is needed to match the left-hand side of a rule and the minimal context that is ``constructed'' by the right-hand side of the rule, respectively)  are tracked.
Intuitively, these concepts are similar to our notions of redex and contractum patterns.
The main differences with respect to our work are as follows. 
First, in \cite{FT94} the slicing is given as a context, while we consider term slices. 
Second, the slice is obtained only on the first term of the sequence by the transitive and reflexive closure of the dependence relation, while  we  slice the whole execution trace, step by step. Obviously, their 
notion of slice is smaller, but we think that our approach can be more useful for trace analysis and program debugging. 
An extension of~\cite{BKV00} is described in~\cite{Terese03book}, 
which provides a generic definition of labeling that works not only for orthogonal TRSs as is the case
of~\cite{BKV00} but for the wider class of all left-linear TRSs. The nonleft-linear case is not handled by \cite{Terese03book}.
Specifically, \cite{Terese03book}~describes a methodology of static and dynamic tracing that is mainly based on the notion of {\em sample of a traced proof term} ---i.e., a pair $(\mu,P)$ that records a rewrite step $\mu = s \rightarrow t$, and a set $P$ of reachable positions in $t$ from  a set of observed positions in $s$.
The  tracing proceeds forward, while ours employs a backward strategy that is particularly convenient for error diagnosis and program debugging. 
Finally, \cite{FT94} and \cite{Terese03book} apply to TRSs whereas  we deal with the richer framework of RWL that considers equations and equational axioms, %
namely rewriting modulo equational theories.

\bibliographystyle{splncs03}

\appendix
\newpage

\section{Proofs of Theorems \ref{prop:reproduded} and 2}\label{app:proof}

\subsection*{Proof of Theorem \ref{prop:reproduded}}\label{app:proof}

We first demonstrate some auxiliary results which 
facilitate the proof of Theorem~\ref{prop:reproduded}.
The following auxiliary result is straightforward.

\medskip
\begin{lemma}\label{le:instance} 
Let $t^\bullet$ be a term slice, and let $t'$ be a term such that $t^\bullet \propto t'$.
For every position 
$w \in \pos(t')$, it holds that, either $root(t'_{|w}) = root(t^\bullet_{|w})$, or there exists a position $u$ of $t^\bullet$ such that $u\leq w$ and $root(t^\bullet_{|u}) = \bullet$. %
\end{lemma}
\begin{proof}
Immediate by Definition \ref{def:app}.
\qed
\end{proof}

The following definitions are auxiliary.
Let $C$ be a context. We define the set of positions of  $C$ as the  set $\pos(C)=\{v\mid root(C_{|v})\neq\Box\}$. 
Given a term $t$, by $path_w(t)$, %
we denote the set of symbols in $t$ that occur in the path from its root to the position $w$ of $t$, e.g., 
$path_{(2.1)}({f(a,g(b),c)})= \{f, g, b\}$.

\begin{definition}
Let $r:\lambda \raw \rho$ be a rule of $\cR$.
Let $\mu:s\stackrel{r,\sigma}{\raw} t$ be a rewrite step such that $s=C[\lambda\sigma]$ and $t=C[\rho\sigma]$.
Given a position $w$, we say that $w$ \emph{is involved in} $\mu$, if there exist $w'$ and $w''$ such that $w=w'.w''$, $C_{|w'} = \Box$ and $w''\in\pos(\rho\sigma)$.
\end{definition}

The following lemma establishes that, if a relevant position is involved in a rewrite step, then the origin position relation preserves the redex pattern of the rule.

\begin{lemma}\label{le:singlepos}
Let $r:\lambda \raw \rho$ be a rule of an elementary rewrite theory $\cR$.
Let $\mu:s\stackrel{r,\sigma}{\raw} t$ be a rewrite step such that $s=C[\lambda\sigma]$ and $t=C[\rho\sigma]$, where $\sigma$ is a substitution
and $C$ is a context.
Let $L$ be a labeling for the rewrite step $\mu$, and $w\in \pos(t)$.

\begin{enumerate}
\item if $w\in\pos(C)$, then $\lhd_\mu^L w=\{v\in\pos(C)\mid w=v.v'\}$
\item if $w=w'.w''$,  $C_{|w'} = \Box$, and $w''\in\pos(\rho\sigma)$, then $\lhd_\mu^L w\supseteq \{w'.v'\in\pos(s)\mid v'\in\pos(\lambda)\}$
\end{enumerate}
\end{lemma}

\begin{proof}
Given the rule $r:\lambda\raw \rho$ and the labeling $L$ for the rewrite step $\mu:s\stackrel{r,\sigma}{\raw} t$, let us consider the labeled rewrite step 
$\mu^L:s^L\stackrel{\hspace{2mm}r^L,\sigma^L}{\raw} t^L$. %
By Definition \ref{def:labelled-step}, we can decompose the labeling $L$ into three labelings $L_C$, $L_r$, and $L_\sigma$ that respectively label the context $C$, the redex and the contractum patterns appearing in $\mu$,  and the terms in $\mu$ introduced by the substitution $\sigma$. 
In other words, we have  $s^L=C^{L_C}[\lambda^{L_r}\sigma^{L_\sigma}]$ and  $t^L=C^{L_C}[\rho^{L_r}\sigma^{L_\sigma}]$. 

Let us prove the two claims independently.

\noindent{\textbf{Claim 1.}} 
We assume that $w\in \pos(t)$ and $w\in\pos(C)$. 
Since the context $C$ has the same initial labeling $C^{L_C}$ in both $s$ and $t$, and 
the sets $Cod(L_C)$, $Cod(L_r)$, and $Cod(L_\sigma)$ are pairwise disjoint, the set of origin positions  $\lhd_{s\raw t}^L w$  in $s$ is the set of positions lying on the path from the root position of $s$ to $w$. Hence, $\lhd_\mu^L w=\{v\in\pos(C)\mid w=v.v'\}$.

\noindent{\textbf{Claim 2.}} We assume that  $w=w'.w''$, $C_{|w'} = \Box$, and $w''\in\pos(\rho\sigma)$. Then, since $r$ belongs to an elementary rewrite theory $\cR$, $r$ is non-collapsing. This implies that there exists a labeled symbol $f^{l'}\in path_w(t^{L})$ belonging to the contractum pattern of the rule $r$. 
By Definition \ref{def:ruleLabel}, for each labeled symbol $g^l$ in the redex pattern of $r$, we have that $l\subseteq l'$. 
Now, since the redex pattern of $r$ is embedded into $s$ and the contractum pattern of $r$ is embedded into  $t$, the inclusion $\lhd_\mu^L w\supseteq \{v.v'\in\pos(s)\mid v'\in\pos(\lambda)\}$ trivially holds by Definition~\ref{def:tracing}.
\qed
\end{proof}

The following lemma establishes that, 
given the rewrite step $\mu: t_0 \stackrel{r}\raw t_1 $ and
a term slice $t^{\bullet}_0$ of $t_0$,    any concretization  of  $t^{\bullet}_0$  is reduced by
the rule $r$  to the corresponding term slice concretization of $t_1$.

\begin{lemma}\label{le:relevant}
Let $r:\lambda\raw \rho$ be a rule of an elementary rewrite theory $\cR$.
Let $\mu:t_0\stackrel{r,\sigma}{\raw} t_1$ be a rewrite step such that $t_0=C[\lambda\sigma]$ and $t_1=C[\rho\sigma]$, where $\sigma$ is a substitution and $C$ is a context.
Let $L$ be a labeling for the rewrite step $\mu$, and $[P_0,P_1]$ be the sequence of  the relevant position sets for 
$\mu:t_0\stackrel{r,\sigma}{\raw} t_1$ w.r.t. the slicing criterion $\cO$. Let $t_0^\bullet=slice(t_0,P_0)$, and  $t_1^\bullet=slice(t_1,P_1)$.
\begin{enumerate}
\item if $P_1\subseteq\pos(C)$ then $t_0^\bullet=t_1^\bullet$.
\item if $P_1\cap\{w|w=v.v'$,  $C_{|v}=\Box$, and $v'\in \pos(\rho\sigma)\}\neq\emptyset$, then for any concretization 
$t_0'$ of $t_0^\bullet$, we have that $t'_0\stackrel{r,\sigma'}{\raw} t'_1$ where  
$t_1^\bullet \propto t_1'$.
\end{enumerate}
\end{lemma}

\begin{proof}
We proof the two claims separately.

\noindent{\textbf{Claim 1.}} Let $P_1\subseteq\pos(C)$. 
Then, by Lemma \ref{le:singlepos} (Claim $1$), for any $w\in P_1$, $\lhd_\mu^L w=\{v\in\pos(C)\mid w=v.v'\}$.
Additionally, by Definition \ref{def:rlvSym}, $P_0= \bigcup_{w\in P_1}( \lhd_\mu^L w)$, and hence $P_0= \bigcup_{w\in P_1}\{v\in\pos(C)\mid w=v.v'\}$.
Therefore, it holds that  
(i) $P_1\subseteq P_0\subseteq\pos(C)$, and for any $v\in P_0\setminus P_1$, there exists a position $v'$ such that $w=v.v'$ for some $w\in P_1$;
(ii) by Definition \ref{def:termSlice}, the function $slice(t,P)$ delivers a term slice $t^\bullet$ where all the symbols of $t$ that do not occur in the path connecting the root position of $t$ with some position $w\in P$ are abstracted by the $\bullet$ symbol. 
Now, since $t_0^\bullet=slice(t_0,P_0)$ and $t_1^\bullet=slice(t_1,P_1)$, 
by (i) and (ii), we can conclude that
$\lambda\sigma$ and $\rho\sigma$ are abstracted by $\bullet$, and the context $C$ is abstracted by the term slice $C^\bullet$ in both $t_0$ and $t_1$. Hence,   $t_0^\bullet=C^\bullet[\bullet]=t_1^\bullet$.

\noindent{\textbf{Claim 2.}}  We assume $P_1\cap\{w|w=v.v'$, $C_{|v}=\Box$,  and $v'\in \pos(\rho\sigma)\}\neq\emptyset$. 
Then, there exists a position $w\in P_1$ such that $w\in \{w|w=v.v'$, $C_{|v}=\Box$, and $v'\in \pos(\rho)\}$. 
By Lemma \ref{le:singlepos} (Claim $2$), it follows that $\lhd_\mu^L w\supseteq \{v.v'\in\pos(t_0)\mid v'\in\pos(\lambda)\}$. 
By Definition \ref{def:rlvSym}, $P_0= \bigcup_{w\in P_1}(\lhd_\mu^L w)$, and hence $P_0\supseteq \{v.v'\in\pos(t_0)\mid v'\in\pos(\lambda)\}$.
Now, by Definition \ref{def:termSlice} and the fact that $P_0\supseteq \{v.v'\in\pos(t_0)\mid v'\in\pos(\lambda)\}$, the redex pattern of the rule $r$ is embedded into $t_0^\bullet=slice(t_0,P_0)$. 
In other words,  $t_0^\bullet=C^\bullet[\lambda\sigma^\bullet]$, where $C^\bullet$ is a term slice for the context $C$, and $\sigma^\bullet$ represents the term slices for the terms introduced by the substitution $\sigma$.  
Thus, by Lemma \ref{le:instance}, any concretization $t_0'$ of $t_0^\bullet$ has the form $ t_0'=C'[\lambda\sigma'] $, where $ C^\bullet \propto C'$ and  for each $x/t\in \sigma'$, there exists $x/t^\bullet\in \sigma^\bullet$  such that $ t^\bullet \propto t$. 
Note also that $t_0^\bullet$ embeds the redex pattern $\lambda^\Box$ of $r$. 
Furthermore, since $r$ belongs to the elementary rewrite theory $\cR$, $r$ is left-linear. 
Thus, the following rewrite step $t_0' \stackrel{r,\sigma'}{\raw} t'_1$ can be executed for any substitution $\sigma'$.
The rewrite step $t_0' \stackrel{r,\sigma'}{\raw} t'_1$ can be decomposed as follows:  $t_0'=C'[\lambda\sigma'] \stackrel{r,\sigma'}{\raw} C'[\rho\sigma']$, for some context $C'$ and substitution $\sigma'$. 
Moreover, by definition of rewrite step, $t'_1$ embeds the contractum pattern of $r$.
Finally, $t_1^\bullet=C^\bullet[\rho^\bullet\sigma^\bullet]$,  and thus $t_1'$ is a concretization of $t_1^\bullet$. 
\qed
\end{proof}

The following proposition allows the soundness of our methodology to be proved for one-step traces on an elementary rewrite theory.

\begin{proposition}\label{pr:onestep}
Let $\cR$ be an elementary rewrite theory. 
Let $\cT$ be an  execution trace in $\cR$, and let $\cO$ be a slicing criterion for $\cT$.
Let $\cT^\bullet : t_0^\bullet \stackrel{r_1}{\rightarrow} t_1^\bullet$ be the trace slice w.r.t.\ $\cO$ of $\cT$.
Then, for any concretization  $t_0'$ of $t_0^\bullet$, 
it holds that $\cT':t_0' \stackrel{r_1}{\rightarrow} t_1'$
is an  execution trace in $\cR$ such that
$t_1^\bullet \propto t_1'$.
\end{proposition}

\begin{proof}
Given the trace slice $\cT^\bullet : t_0^\bullet \stackrel{r_1}{\rightarrow} t_1^\bullet$ w.r.t.\ $\cO$ of $\cT$,
let $[P_0,P_1]$ be the sequence of the relevant position sets of $\cT$ w.r.t.\ $\cO$.
We have (i) $t_0^\bullet=slice(s_0,P_0)$ and $t_1^\bullet=slice(s_1,P_1)$, where  $s_0\stackrel{r_1}{\rightarrow} s_1$ is a rewrite step occurring in $\cT$; 
(ii) $t_0^\bullet\neq t_1^\bullet$.  
Let $r_1$ be the  rule $\lambda\raw \rho$.
The rewrite step $s_0\stackrel{r_1}{\rightarrow} s_1$ can be decomposed as follows:  $s_0=C[\lambda\sigma]\stackrel{r_1}{\rightarrow} C[\rho\sigma]=s_1$, for some context $C$ and substitution $\sigma$. 

Since $\cR$ is elementary and $t_0^\bullet\neq t_1^\bullet$,  by Claim $1$ of Lemma \ref{le:relevant}, $P_1\not\subseteq\pos(C)$. Hence, there exists a position $w\in P_1$ such that $w=v.v'$ and $v'\in\pos(\rho\sigma)$. 
Also, because $\cR$ is elementary, we can apply   Claim $2$ of Lemma~\ref{le:relevant}, and for any concretization $t_0'$ of $t_0^\bullet$, we get $t_0'\stackrel{r_1}{\rightarrow} t_1'$ such that $t_1'$ is a concretization of $t_1^\bullet$.
\qed
\end{proof}

\noindent{\textbf{Theorem \ref{prop:reproduded}.}} \textit{(soundness)}
\textit{
Let $\cR$ be an elementary rewrite theory. Let $\cT$ be an  execution trace in $\cR$ and let $\cO$  be a slicing criterion for $\cT$.
Let $\cT^\bullet : t_0^\bullet \stackrel{r_1}{\rightarrow} t_1^\bullet  \ldots \stackrel{r_n}{\rightarrow} t_n^\bullet $ be the corresponding  trace slice w.r.t.\ $\cO$. 
Then, for any concretization  
$t_0'$ of $t_0^\bullet$, 
it holds that $\cT':t_0' \stackrel{r_1}{\rightarrow} t_1' \ldots \stackrel{r_n}{\rightarrow} t_n'$
is an  execution trace in $\cR$, and   $t_i^\bullet \propto t_i'$, for $i=1,\ldots,n$.
}

\begin{proof}
The proof proceeds by induction on the length of the trace slice $\cT^\bullet$ and exploits Proposition \ref{pr:onestep} to prove the inductive case.  
Routine.
\qed
\end{proof}

\subsection*{Proof of Theorem 2}
\label{app:proof-extended}

In oder to prove Theorem~2, we use the same proof scheme as for elementary rewrite theories, since the extended technique described in Section~\ref{sec:slicing-m-e} is only concerned with suitable extensions of the labeling procedure given in Definition~\ref{def:labelled-step}, which do not affect the overall backward trace slicing methodology.

Let us start by proving an extension of Lemma~\ref{le:singlepos} (Claim 2), which holds for nonleft-linear as well as collapsing rules.

\begin{lemma}\label{le:singlepos-ext}
Let $r:\lambda \raw \rho$ be a rule that is either nonleft-linear or collapsing.
Let $\mu:s\stackrel{r,\sigma}{\raw} t$ be a rewrite step such that $s=C[\lambda\sigma]$ and $t=C[\rho\sigma]$, where $\sigma$ is a substitution
and $C$ is a context.
Let $L$ be a labeling for the rewrite step $\mu$, and $w\in \pos(t)$. Then,

\begin{enumerate}
\item  if $w\in\pos(C)$, then $\lhd_\mu^L w=\{v\in\pos(C)\mid w=v.v'\}$
\item if $w=w'.w''$,  $C_{|w'} = \Box$, and $w''\in\pos(\rho\sigma)$, then $\lhd_\mu^L w\supseteq \{w'.v'\in\pos(s)\mid v'\in\pos(\lambda)\}$
\end{enumerate}

\end{lemma}

\begin{proof}
We prove the two claims separately.

\noindent{\bf Claim 1.} The proof is identical to the proof of Claim 1 of  Lemma~\ref{le:singlepos}.

\noindent{\bf Claim 2.}
To prove the lemma, we distinguish three cases.

\begin{description}
\item[{\bf Case 1: Rule $r$ is collapsing.}]

Given the collapsing rule $r=\lambda\raw \rho$ where $\rho = x$ with  $x\in\Var(\lambda)$, let us consider  the term $t_i$  introduced by the substitution $\sigma$ via the binding $x/t_i$, and 
we have $\mu=C[\lambda\sigma] \stackrel{r}{\rightarrow} C[t_i]$.
Let us also consider the labeled rewrite step $\mu^L:s^L\stackrel{\hspace{2mm}r^{L_r},\sigma^{L_\sigma}}{\raw} t^L$ via the labeling $L$.
By Definition~\ref{def:labelled-step}, we have  $s^L=C^{L_C}[\lambda^{L_r}\sigma^{L_\sigma}]$ and $t^L=C^{L_C}[t_i^{L_\sigma}]$.

Let $f^{l'}$ be the labeled root symbol of $t_i^{L_\sigma}$.
By Definition~\ref{def:collapsext} (Step $s_4$), we have that $l' = l_\lambda l_i$, where $l_\lambda$ is formed by joining all the labels appearing in the redex pattern $\lambda^{L_r}$ and $l_i$ is the label of the root of the labeled term $t_i^{L_\sigma}$. 
This implies that, for each labeled symbol $g^l$ in the redex pattern of $r$, we have that $l \subseteq l'$. 
Furthermore, by hypothesis, we have that $w\in C[t_i]$ and  $w'' \in Pos(t_i)$. 
Hence, by Definition~\ref{def:tracing}, the inclusion $\lhd_\mu^L w\supseteq \{v.v'\in\pos(s)\mid v'\in\pos(\lambda)\}$ trivially holds.

\item[{\bf Case 2: rule $r$ is nonleft-linear.}]
Given the nonleft-linear rule $r$, the proof is perfectly analogous to the proof of Lemma~\ref{le:singlepos} since, by Definition~\ref{def:nllext} (Step $s_5$), the label of each symbol in the contractum pattern of the rule $r$  includes all the labels appearing in the redex pattern of $r$.   

\item[{\bf Case 3: rule $r$ is collapsing and nonleft-linear.}] Since $r$ is both collapsing and nonleft-linear, $\mu$ is labelled  according to Definition~\ref{def:collapsext} (Step $s_4$) and Definition~\ref{def:nllext} (Step $s_5$). Therefore, we can prove the claim by simply combining the arguments used to prove Case $1$ ad Case $2$. 
\end{description}
\qed
\end{proof}

The following Lemma extends Lemma~\ref{le:relevant} to deal with collapsing and nonleft-linear rules.

\begin{lemma}\label{le:relevant-extended}
Let $r:\lambda\raw \rho$ be a rule which is either left-linear or collapsing.
Let $\mu:t_0\stackrel{r,\sigma}{\raw} t_1$ be a rewrite step such that $t_0=C[\lambda\sigma]$ and $t_1=C[\rho\sigma]$, where $\sigma$ is a substitution and $C$ is a context.
Let $L$ be a labeling for the rewrite step $\mu$, and $[P_0,P_1]$ be the sequence of  the relevant position sets for 
$\mu:t_0\stackrel{r,\sigma}{\raw} t_1$ w.r.t. the slicing criterion $\cO$. Let $t_0^\bullet=slice(t_0,P_0)$, and  $t_1^\bullet=slice(t_1,P_1)$. Then,
\begin{enumerate}
\item if $P_1\subseteq\pos(C)$ then $t_0^\bullet=t_1^\bullet$.
\item if $P_1\cap\{w|w=v.v'$, $C_{|v}=\Box$, and $v'\in \pos(\rho\sigma)\}\neq\emptyset$, then for any concretization $t_0'$ of $t_0^\bullet$, we have that $t'_0\stackrel{r,\sigma'}{\raw} t'_1$ where  
$t_1^\bullet \propto t_1'$.
\end{enumerate}
\end{lemma}

\begin{proof}
We proof the two claims separately.\\

\noindent{\bf Claim 1.} The proof is identical to the proof of Claim 1 of  Lemma~\ref{le:relevant}.

\noindent{\bf Claim 2.}
To prove the lemma, we distinguish three cases.

\begin{description}
\item[{\bf Case 1: rule $r$ is collapsing.}]
Given the collapsing rule $r$, the proof is perfectly analogous to the one of Lemma \ref{le:relevant} Claim 2. By using Lemma \ref{le:singlepos-ext} instead of Lemma~\ref{le:singlepos}, we 
are still able to prove that the redex pattern of $r$ embedded in $t_0$ is also embedded in $t_0^\bullet$, and hence  for any concretization $ t_0'$ of  $t_0^\bullet$, the rewrite step $t'_0\stackrel{r,\sigma'}{\raw} t'_1$ can be proved.
Finally, by using the same argument of  Lemma \ref{le:relevant} Claim 2, we conclude that $t_1^\bullet \propto t_1'$.

\item[{\bf Case 2: rule $r$ is nonleft-linear.}]
Given the nonleft-linear rule $r$, the proof is similar to the one of Lemma~\ref{le:relevant}. By exploiting Lemma \ref{le:singlepos-ext} and  Definition~\ref{def:nllext} (Step $s_5$), we can show that (i) the redex pattern of $r$ embedded in $t_0$ is also embedded in $t_0^\bullet$, and (ii) for each term $t$ introduced in $t_0$
by a binding $x/t\in\sigma$ such that $x$ occurs multiple times  in $\lambda$, $t$ is preserved in $t_0^\bullet$ (i.e., $t$ is not abstracted by $\bullet$ in $t_0^\bullet$).
By (i) and (ii), it is immediate to prove that, for any concretization $t_0'$ of $t_0^\bullet$, the rewrite step $t'_0\stackrel{r,\sigma'}{\raw} t'_1$ can be proved. Finally,  by using the same argument of Lemma \ref{le:relevant} Claim 2, we can show that $t_1^\bullet \propto  t_1'$.

\item[{\bf Case 3: rule $r$ is collapsing and nonleft-linear.}] Firstly we observe that, as the rule $r$ is collapsing, by Lemma~\ref{le:singlepos-ext} 
the redex pattern of $r$ embedded in $t_0$ is also embedded in $t_0^\bullet$, and hence  for any concretization $t_0'$ of  $t_0^\bullet$,
the redex pattern of $r$ is embedded in $t_0'$ as well. Secondly, since $r$ is nonleft-linear, by Lemma \ref{le:singlepos-ext} and  Definition~\ref{def:nllext} (Step $s_5$), for each term $t$ introduced in $t_0$
by a binding $x/t\in\sigma$ such that $x$ occurs multiple times  in $\lambda$, $t$ is preserved in $t_0^\bullet$. Hence, $t$ is also embedded in $t_0'$, for any concretization $t_0'$ of $t_0^\bullet$.
From the two facts above, it directly follows that for any $t_0'$ such that $t_0^\bullet\propto t_0'$, the rewrite step $t'_0\stackrel{r,\sigma'}{\raw} t'_1$ can be proved. Finally,  by using the same argument of Lemma \ref{le:relevant} Claim 2, we can show that $t_1^\bullet \propto  t_1'$.
  
\end{description}	
\qed \end{proof}

The following proposition allows us to prove the soundness of our methodology for one-step traces on an extended rewrite theory. 

\begin{proposition}\label{pr:onestep-ext}
Let $\cR$ be an extended rewrite theory. Let $\cT: t_0 \stackrel{r_1}{\rightarrow} t_1$ be an  execution trace in $\cR$, and let $\cO$ be a slicing criterion for $\cT$.
Let $\cT^\bullet : t_0^\bullet \stackrel{r_1}{\rightarrow} t_1^\bullet$ be the trace slice w.r.t.\ $\cO$ of $\cT$.
Then, for any concretization  $t_0'$ of $t_0^\bullet$, 
it holds that $\cT':t_0' \stackrel{r_1}{\rightarrow} t_1'$
is an  execution trace in $\cR$ such that %
$t_1^\bullet \propto t_1'$.
\end{proposition}

\begin{proof}
Consider the rewrite step  $\mu: t_0 \stackrel{r_1}{\rightarrow} t_1$. 
In the case when $r_1$ is  left-linear and non-collapsing (i.e., a rule belonging to an elementary rewrite theory), the proof is identical to the proof of Proposition~\ref{pr:onestep-ext}.  
Hence w.l.o.g. we assume that $r$ corresponds to  a
 collapsing or nonleft-linear rule,  built-in operator evaluation, or AC axiom.  

\begin{description}
\item[Nonleft-linear/collapsing rules.] In this case, the proof of Proposition~\ref{pr:onestep-ext} is analogous to the proof of Proposition \ref{pr:onestep}, by using Lemma~\ref{le:relevant-extended} in the place of Lemma~\ref{le:relevant}. 

\item[Built-in Operators.]
Let $t_0=C[op(t_1, \ldots ,t_m)]$ and $t_1=C[t']$. Hence, $\mu : C[op(t_1, \ldots ,t_m)] \raw C[t']$ is a rewrite step mimicking the evaluation of the built-in operator call $op(t_1, \ldots ,t_m)$.
By Definition~\ref{def:nbuilt-in} and Definition \ref{def:tracing}, it is immediate to show that $op(t_1, \ldots ,t_m)$ is embedded
in $t_0^\bullet$, and thus for any concretization  $t_0^\bullet \propto t_0'$, $t_0' \stackrel{r_1}{\rightarrow} t_1'$ and $t_1^\bullet \propto  t_1'$.

\item[\bf Associative-Commutative Axioms.]
Flat/unflat transformations are interpreted as rewrite steps that reduce AC symbols.
Let us first consider the flat transformation  $t\raw_{flat_B}t'$ that reduces the AC symbol $f$. 
By Definition~\ref{def:nAC}, the label of the occurrence of $f$ in  $t'$ contains all the labels of the different occurrences of $f$ appearing in $t$ that have been reduced by the transformation. In other words, the label of $f$ in $t'$ keeps track of all the occurrences of $f$ that have been reduced in $t$, and therefore the claim holds directly. 
The claim for unflat transformations can be proved in a similar way. 
\end{description}
\qed
\end{proof}

Finally, we exploit Proposition \ref{pr:onestep-ext} in order to prove the extended soundness of our methodology on extended rewrite theories.

\medskip

\noindent{\textbf{Theorem 2.}} \textit{(extended soundness)}
\textit{
Let $\cR = (\Sigma, E, R)$ be an extended rewrite theory. 
Let $\cT$ be an  execution trace in the rewrite theory $\cR$, and let $\cO$ be a slicing criterion for $\cT$.
Let $\cT^\bullet : t_0^\bullet \stackrel{r_1}{\rightarrow} t_1^\bullet  \ldots \stackrel{r_n}{\rightarrow} t_n^\bullet $ be the corresponding  trace slice w.r.t.\ $\cO$. 
Then, for any concretization  
$t_0'$ of $t_0^\bullet$, 
it holds that $\cT':t_0' \stackrel{r_1}{\rightarrow} t_1' \ldots \stackrel{r_n}{\rightarrow} t_n'$
is an  execution trace in $\cR$ and   $t_i^\bullet \propto t_i'$, for $i=1,\ldots,n$.
}

\begin{proof}
The proof proceeds by induction on the length of the trace slice $\cT^\bullet$ and exploits Proposition \ref{pr:onestep-ext} in order to prove the inductive case.  
 Routine.
\qed
\end{proof}

\end{document}